\documentclass[conference]{IEEEtran}
\IEEEoverridecommandlockouts

\usepackage{cite}
\usepackage{amsmath,amssymb,amsfonts}
\usepackage{graphicx}
\usepackage{textcomp}
\usepackage{xcolor}
\usepackage{comment}

\usepackage{physics,algorithm,algpseudocode}
\usepackage{amsthm}
\newtheorem{theorem}{Theorem}
\newtheorem{corollary}{Corollary}[theorem]

\def\BibTeX{{\rm B\kern-.05em{\sc i\kern-.025em b}\kern-.08em
    T\kern-.1667em\lower.7ex\hbox{E}\kern-.125emX}}

\begin{document}

\title{Quantum Network Routing based on Surface Code Error Correction\\
}


\author{
\IEEEauthorblockN{Tianjie Hu, Jindi Wu, and Qun Li}
\IEEEauthorblockA{\text{
Department of Computer Science, 
William \& Mary, 
Williamsburg, VA 23185, USA}
}}

\maketitle

\thispagestyle{plain}
\pagestyle{plain}

\newcommand{\RR}{\mathbb{R}}
\newcommand{\CC}{\mathbb{C}}
\newcommand{\dsp}{\displaystyle}




\begin{abstract}
\textbf{ } Quantum networks encounter unavoidable channel noises and erasure errors, presenting a huge obstacle in designing protocols that attain both high reliability and efficiency. Typically, quantum networks fall into two categories: 
those utilize quantum entanglements for quantum teleportation, and those directly transfer the actual quantum messages.
In this paper, we present SurfNet, a quantum network that inherits the main advantages from both categories. It employs surface codes as logical qubits for encoding messages, and utilizes two parallel communication channels to fault-tolerantly transfer each surface code in a modular manner. 
Our approach of using surface codes can timely correct both operational and photon loss errors within the network, and the integration of the two channels within the network can greatly improve network throughput. 
For the implementation of SurfNet, we propose a novel network architecture, designed to better integrate surface codes into quantum networks. We also propose a novel error correction decoder, designed to fully utilize the modular characteristic of surface codes within our network. 
Simulation results demonstrate that SurfNet with its decoder significantly enhances the communication fidelity within quantum networks.
\end{abstract}


\section{Introduction}
Quantum networks find applications in various domains such as distributed quantum machine learning, quantum cryptography, quantum key distribution, and more~\cite{nielsen2002quantum, tengner2008photonic, wu2022scalable, wu2024distributed}. 
Due to the intricacies of quantum mechanics, designing quantum networks requires carefully managing both reliability and efficiency. It has been a prominent area of research for several decades, with numerous research groups proposing various quantum network designs. These designs can be broadly classified into two schemes based on how they transmit quantum messages within the network.

The first quantum network scheme~\cite{briegel1998quantum, zeng2022multi, shi2020concurrent, qiao2022quantum, chen2022heuristic, farahbakhsh2022opportunistic, li2022fidelity, zhao2022e2e} utilizes quantum entanglements for quantum teleportation across nodes. As quantum message is inherently unstable and highly susceptible to environmental noises, in this scheme, a series of quantum entangling and swapping operations are performed among intermediate nodes between the sender and receiver. The sender and receiver will then attain a pair of entangled qubits, which can be utilized for quantum teleportation to teleport the quantum message. 
However, in practice, this scheme faces inefficiency issues due to the short lifespan of entangled pairs. Even worse, generation of entanglements across adjacent nodes is a probabilistic process, which significantly hinders distant quantum teleportation, making it highly time-consuming. In addition, a large amount of classical communications across nodes are required for both entanglement swapping and quantum teleportation, frequently constituting a large portion of the total communication time.

The second quantum network scheme~\cite{fowler2010surface, muralidharan2014ultrafast, hu2023surfacenet} transfers the actual quantum message through optical fibers, and often utilizes logical qubits to carry and preserve the message. In this scheme, the actual transmission instead of teleportation can effectively address the inefficiency observed in the previous scheme, and logical qubits are utilized to address the instability of quantum message.
A logical qubit is typically represented by multiple physical qubits, and its state is determined collectively by these physical qubits. This encoding provides resilience against random noise, through its inherent entanglement and redundancy.
However, it comes with the drawback of significantly increasing the traffic in networks since to transmit a single logical qubit, multiple physical qubits must be transmitted.

\begin{figure}[t]
\centerline{\includegraphics[width= \linewidth]{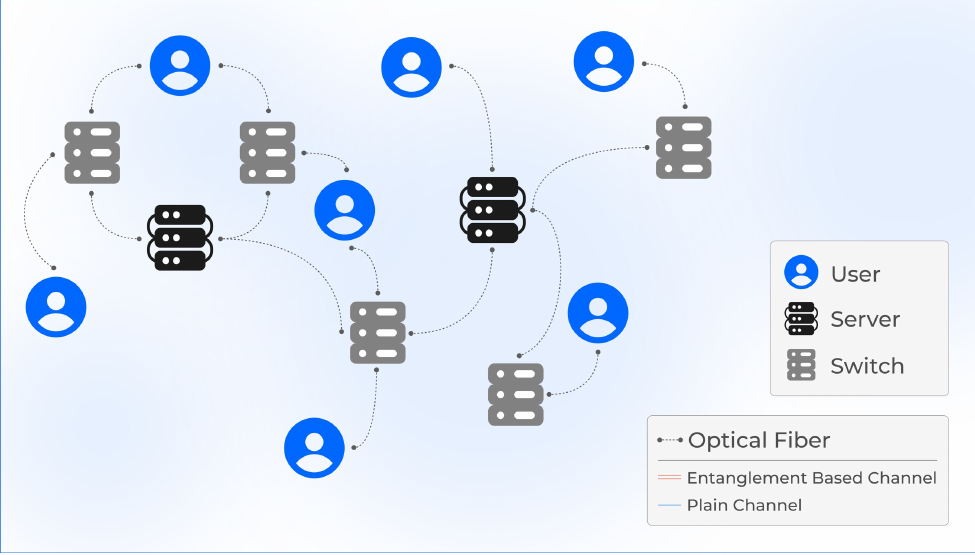}} 
\caption{\text{SurfNet quantum network.} Users communicate with each other through switches and servers, interconnected by entanglement-based channels and plain channels within optical fibers.}
\label{fig:network_3}
\end{figure}

In this paper, we present \textit{SurfNet}, a new quantum networking strategy that combines techniques from both quantum network schemes above. We aim to simultaneously achieve high network throughput and communication fidelity. Throughout this paper, \textit{throughput} is defined as the number of concurrent communications, and \textit{fidelity} is defined as the success rate of each communication. 
SurfNet employs surface codes as logical qubits for encoding communication messages, to effectively cope with errors induced by noise within networks. The choice of surface codes is motivated by their compact code size and high error threshold. Incorporating them into quantum networks offers a solution to promptly address operational errors and photon losses within the network via timely performing error corrections on surface codes.

In the meantime, to minimize traffic within networks, SurfNet chooses the surface codes with low code distances, which have fewer physical qubits to be transmitted but are normally with higher logical error rate. However, in SurfNet, we are able to lower down their logical error rates with the dual-channel design within SurfNet. In specific, each surface code in SurfNet is transferred in a modular manner as two parts: the \textit{Core} part, consisting of qubits critical to the logical error rate, and the \textit{Support} part, containing the remaining qubits which are less critical but still essential for error corrections. These two parts are transmitted separately using two parallel communication channels: the \textit{entanglement-based channel} that transfers the Core part utilizing quantum teleportation, and the \textit{plain channel} that transfers the Support part as photons. This dual-channel approach strikes for a balance between the reliability and efficiency of quantum networks, and addresses the limitations posed by the low entanglement generation rates. As demonstrated later in the routing protocol design, a coordinated use of both channels can effectively achieve high network throughput and communication fidelity.

We also present a new error correction decoder specifically designed for SurfNet. It is based on the finding that physical qubits within the same surface code may experience distinct error rates, and exploiting these variations can improve error correction accuracy. In SurfNet, this finding is essential since the error rates are highly different at the Core and Support parts on each surface code. Thus, we design the \textit{SurfNet Decoder} to effectively incorporate these distinct error rates with the observed syndromes for efficiently decoding the most probable error pattern. In this paper, we focus solely on the Pauli errors and erasure errors, and decoherence errors are handled separately via error mitigation techniques at each node within the network. For the ease of discussion, w e assume perfect measurements, but we note that our analysis can be readily generalized to include any type of quantum errors, by updating the error rate on each physical qubit.

Our main contributions are summarized as follows:

\begin{itemize}
\item We present the first dual-channel quantum network, incorporating both the entanglement-based channel and the plain channel. The end-to-end communication procedure described in this paper efficiently utilizes both channels for transferring surface codes with low error rate.
\item We present a novel error correction decoder that analyzes surface codes in a modular manner and performs the decoding accordingly. We evaluate its performance against the popular Union-Find decoder. We also illustrate how conventional error correction decoders can be adapted for SurfNet. 
\item We implement the routing protocol of SurfNet as an integer programming problem. We evaluate its performance against a baseline model and the mainstream entanglement-based networks.
\end{itemize}

\section{Related Work}

Quantum networks have emerged as a promising solution for connecting distant quantum devices. Current physical quantum networks~\cite{liao2018satellite, diamanti2016practical, Wehner2018QI} are still in their early stages, constrained by the current capabilities of qubit processing and storage, as well as the high costs associated with establishing photonic channels and repeater nodes.
Efforts have been directed towards leveraging entanglement between nodes for quantum teleportation~\cite{bennett1992communication}. Specifically, techniques such as entanglement purification~\cite{li2022fidelity, zhao2022e2e} have been employed to enhance the quality of entangled pairs shared across inferior or distant optical fibers.

Extensive research has been conducted in the field of quantum network routing~\cite{zeng2022multi, shi2020concurrent, qiao2022quantum, chen2022heuristic, farahbakhsh2022opportunistic, panigrahy2023scalable, yu2021protocols, gu2023fendi}. Various research groups are actively exploring general architecture designs for quantum networks, including repeaters or routers~\cite{munro2015inside, vasantam2022throughput, lee2022quantum}, as well as the development of photonic channels~\cite{zhou2023simulator, sun2016quantum, de2022quantum}. 
In the realm of fault tolerant  communication, efforts have been made in~\cite{fowler2010surface, muralidharan2014ultrafast, hua2021autobraid, hu2023surfacenet}. For surface code, different decoders~\cite{fowler2013minimum, higgott2023sparse, delfosse2021almost} have been proposed to enhance the error threshold of surface codes, thereby reducing the fidelity requirements for optical fiber constructions. These advancements contribute to the overall efficiency and reliability of quantum networks.

\section{Background}

\subsection{Qubits}

Different from classical computing bits, which can only be either 0 or 1, qubits in quantum computing exist as linear combinations of the basis states $\ket{0}$ and $\ket{1}$. This property, known as superposition, allows a qubit to exist in a continuous space between these two basis states, commonly visualized as the Bloch sphere. The specific value of a qubit as either $\ket{0}$ or $\ket{1}$ is only determined upon a quantum measurement. Exploiting superposition, each qubit can simultaneously represent both classical values (0 and 1) of information. Moreover, when $n$ qubits are employed together as an input, this advantage grows exponentially, enabling them to encode $2^n$ pieces of classical information simultaneously. However, the inherent nature of qubits being in superposition also makes them highly susceptible to instability~\cite{google2023suppressing, kim2023evidence}, vulnerable to environmental noise, and prone to random collapse. 
Common errors affecting single qubits are Pauli errors related to the three quantum Pauli gates $\{X,Y,Z\}$. Each of the three Pauli gates flips a qubit based on the corresponding $X/Y/Z$-axis on the Bloch sphere. And Pauli errors are defined as situations where random Pauli gates are unintentionally applied to qubits, often due to environmental noise or crosstalk within circuits. 
Thus, instead of directly utilizing the unstable physical qubits, numerous papers have proposed designs for quantum code designs. These quantum codes are composed of clusters of physical qubits and function as logical qubits~\cite{gottesman1997stabilizer}.

\subsection{Surface Codes}

\begin{figure}[t]
\centerline{\includegraphics[scale=.4]{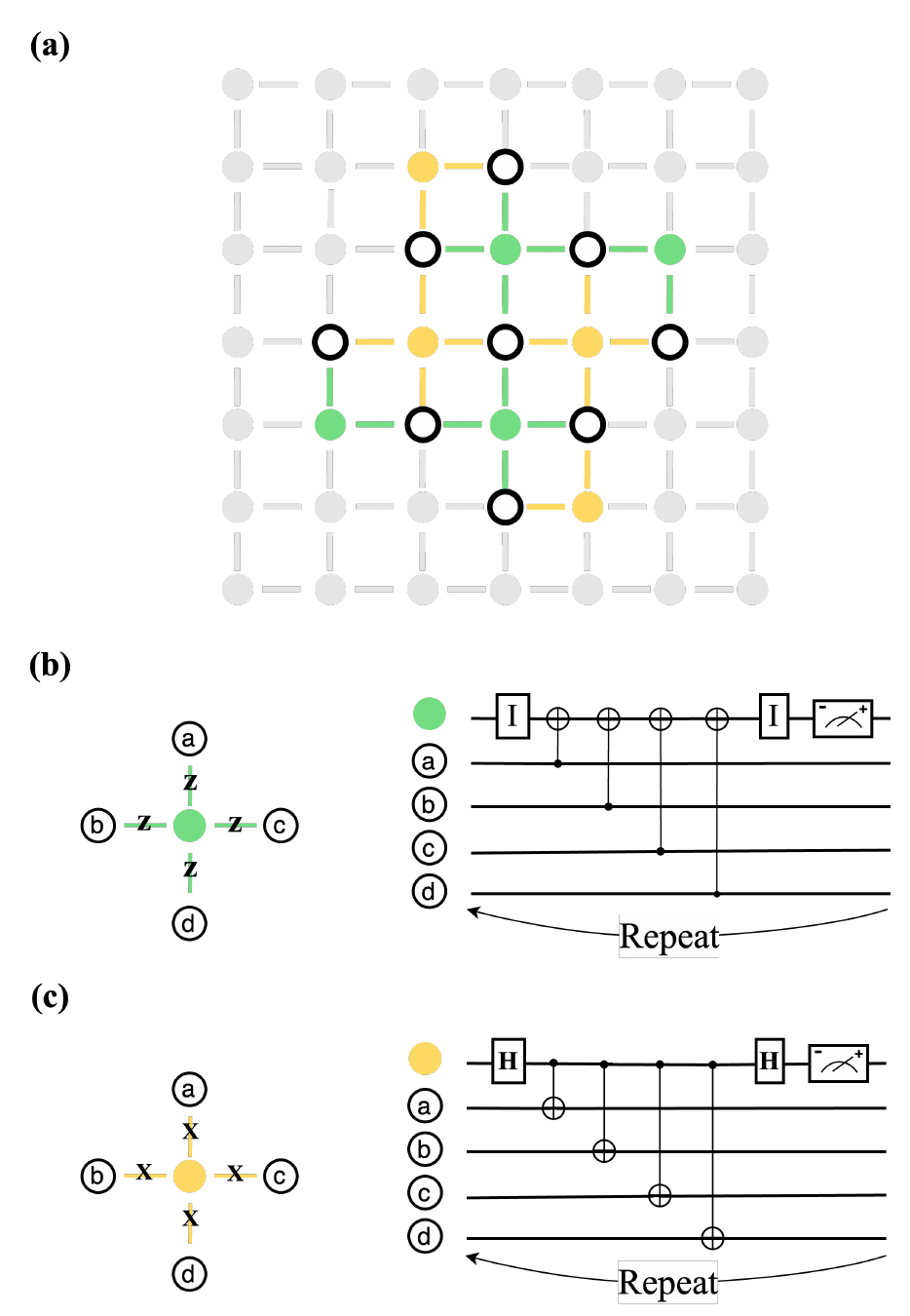}}
\color{black}
\caption{ (a) \textbf{Example distance-3 surface code}, for which distance refers to the minimum number of data qubits required to perform a logical operation. Open circles represent data qubits, and colored dots represent measurement qubits (green for measure-Z, yellow for measure-X). (b)(c) Quantum circuits of measure-Z and measure-X.}
\label{fig:sfcode}
\end{figure}


Among various quantum error correction codes, surface code stands out as the most promising one due to its compact size and high error threshold. Illustrated in Fig.~\ref{fig:sfcode}(a), each surface code consists of two groups of physical qubits: data qubits and measurement qubits. Both groups are essentially interconnected physical qubits through quantum gates, yet they play distinct roles within the surface code. The logical qubit information is encoded and stored within data qubits, while measurement qubits are utilized to preserve information integrity. Each data qubit is connected to its neighboring measurement qubits, comprising two measure-X qubits and two measure-Z qubits, except on the boundaries. And if an error occurs in a specific data qubit, its neighboring measurement qubits would reflect the error: measure-X qubits reflect Y or Z types errors, and measure-Z qubits reflect X or Y types errors. Hence, measurement qubits are also known as syndrome qubits or, in code theory, ancilla qubits.

Serving as logical qubits within quantum networks, each surface code encodes a single logical qubit carrying a one-qubit quantum state. There also exist variants on surface codes such as X-cut, Z-cut, or multi-cuts surface codes~\cite{fowler2012surface} that can encode multiple logical qubits in a single topology, but for the purpose of this paper, we focus on the use of surface codes serving as single logical qubits. 
The mechanism underlying surface codes relies on quantum entanglements and reduction on the degrees of freedom. For a comprehensive introduction into Quantum Computation and Quantum Information, we refer readers to~\cite{nielsen2002quantum}, and for a detailed explanation on Surface Codes,~\cite{fowler2012surface} is recommended.


\begin{figure}[t]
\centerline{\includegraphics[scale=.5]{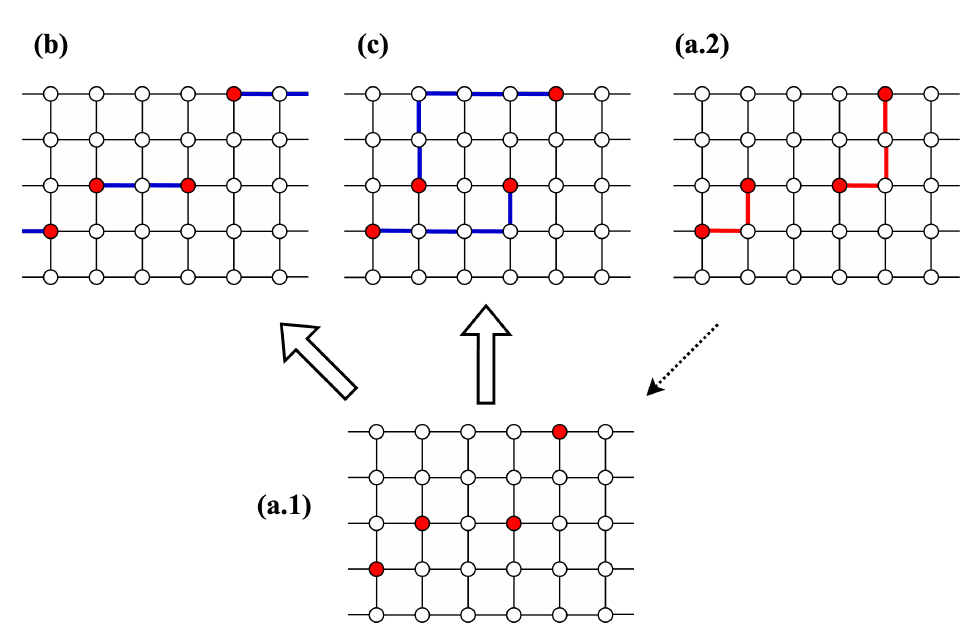}}
\caption{\textbf{Example Surface Code Error Correction.} (a.1) Example syndrome pattern of a surface code, induced by the (a.2) example error pattern. The edges in red represent data qubits with errors, and the vertices in red represent their induced syndromes. (b)(c) Example decoding results from decoders, where (c) is equivalent to (a.2) upon stabilizers, and (b) is not. The edges in blue represent potential erroneous data qubits indicated by decoders.}
\label{fig:decoder}
\end{figure}

\subsection{Error Correction}

Error correction can be performed on surface codes to effectively mitigate the environmental noise. During initialization, each surface code begins with only the data qubits. As we commonly initialize surface codes in their logical $\ket{0}$, all data qubits are commonly prepared in the state $\ket{0}$. Next, we put on the measurement qubits and connect them to their neighboring data qubits using quantum CNOT gates, as depicted in Fig.\ref{fig:sfcode}(b)(c). For the ease of syndrome readout, all measurement qubits are commonly prepared in their $\ket{0}$ states.
After all the CNOT gates are established, we need to perform at least one round of error correction cycle to complete the setup of a surface code. As in Fig.\ref{fig:sfcode}(b)(c), each cycle involves a sequence of CNOT operations and a measurement on each measurement qubit, which forces its neighboring data qubits to be in a mutual eigenstate based on its initial measurement outcome. 
For example, we initialize all data qubits in $\ket{0}$, and via the circuit in Fig.\ref{fig:sfcode}(b), the measurement outcome of each measure-Z qubit is determined to be $\ket{0}$. So we call the state $\ket{0000}$ of qubits abcd to be a valid eigenstate in the $+1$ eigenspace of the measure-Z qubit, for which $+1$ is associated with the $\ket{0}$ measurement outcome. 
However, the initialized state $\ket{0000}$ is not a valid eigenstate for the measure-X qubit as shown in Fig.\ref{fig:sfcode}(c). Thus, the measurement outcome of each measure-X qubit can be either $\ket{0}$ or $\ket{1}$. Yet, once the outcome of each measure-X qubit is determined, its neighboring data qubits are forced to be in one of the possible eigenstates measured in $\ket{+}$'s or $\ket{-}$'s, which are the two basis states along the X-axis of the Bloch sphere.

At this point, we complete the initialization of a surface code, with all data qubits and measurement qubits becoming highly entangled. We name the result of the first round measurement outcomes of all measurement qubits as the quiescent state. 
And if no error occurs in the future, all measurement results in subsequent error correction cycles would stay in the quiescent state. 
If the measurement outcome of any measurement qubit changes, then at least one of its neighboring data qubits has an error, and we mark this measurement qubit as a \textit{syndrome}. In each error correction cycle, all syndromes collectively form a syndrome pattern, as in Fig.~\ref{fig:decoder}(a.1), which is then inputted into an error correction decoder. The decoder aims to connect all syndromes either into pairs or to the boundary, and any data qubit along the connected path will be identified as erroneous. Due to the nature of qubits, the real error pattern, as in Fig.~\ref{fig:decoder}(a.2), cannot be directly observed, since measuring the data qubits would destroy their encoded logical information and collapse them to basis states. 

This decoding problem can be approximated as a minimum weight perfect matching problem, which has been widely studied. Each minimum weight perfect matching decoder connects syndromes with the minimum sum of weights. In our example of Fig.~\ref{fig:decoder}, assuming the weight on each edge is 1, the decoder is likely to output either (a.2) or (b), both of which have a sum of weights being 5. However, these two error patterns differ by a logical operator of the surface code, as the combination of these two patterns traverse across the decoding graph. Meanwhile, the error patterns (a.2) and (c) are equivalent in surface codes, as their combination forms a cycle within the decoding graph. Thus, if the decoder outputs (b), a logical error will arise in this surface code; and if the decoder outputs (a.2) or (c), the errors in this surface codes can be successfully corrected. Other decoding methods, such as Union-Find decoders, tensor network decoders, and machine learning models, have also been proposed, each demonstrating competitive performances in different noise models. Error correction cycles are typically scheduled in a regular basis, and more frequently when the noise level is high.

\section{Network based on Error Correction}

Contemporary error correction decoders suffer from logical errors, which arise when logical operators are implicitly performed by the decoder. For example, if there are only two syndromes observed on a surface code, the decoder have two choices: connecting them into a pair, or separately to two opposite boundaries. These two choices differ by a logical operator, and there is one and only one of the two choices that is equivalent to the real error pattern. Thus, a logical operator would be implicitly applied when the decoder makes the wrong choice.


Theoretically, to prevent logical errors within a surface code, we need to ensure there exists at a ``perfect" qubit on each logical operator that has a $0\%$ physical error rate. This prevents each logical operator from being implicitly applied by the decoder. However, in practice, we can never guarantee a single data qubit to be perfect, but to achieve a similar effect, we can maintain a considerably low error rate for such a qubit. Thus, though we cannot prevent logical errors within a surface code, we are able to reduce the logical error rate of a surface code by selecting a few data qubits and maintaining them with low error rates. 
We define these selected data qubits to be the \textbf{Core part} of each surface code, which is critical to the logical error rate of the surface code; and the remaining data qubits constitute the \textbf{Support part} of each surface code, which is less critical to logical error rate, but is still essential for decoding syndromes. 
The specific selection of data qubits for the Core part depends on a variety of factors such as the required error threshold and the quality of service within networks, which will be discussed in detail in our future work. For the purpose of this chapter, we employ a fixed topology for the Core part, as in Section~\ref{err_decoder}.

\begin{figure*}[t]
\centering
\includegraphics[width= \linewidth]{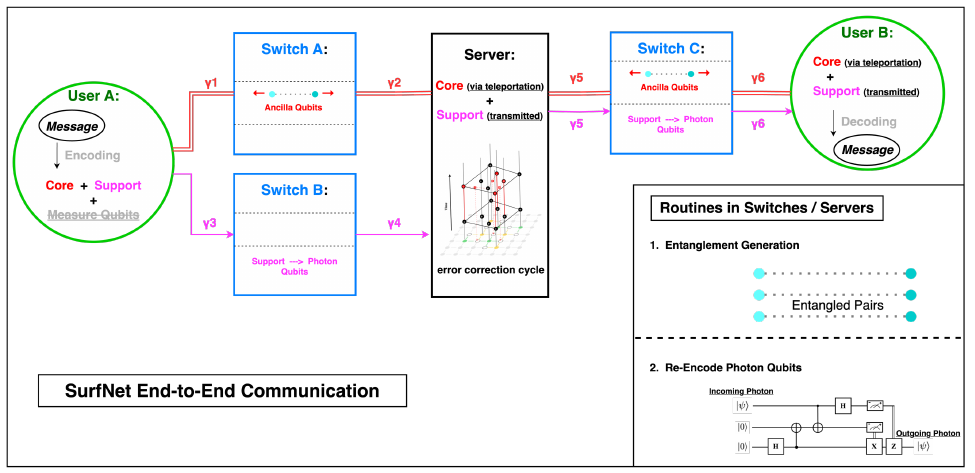}
\caption{\text{Example one-way communication.} A surface code is transferred from user A to B, where the green circles represent Users, blue squares represent Switches, black square represents Server, red double lines represent Optical Fibers for transmitting ancillary entangled pairs, and pink arrows represent Optical Fibers for transmitting the Support parts as photon qubits. Each optical fiber is labeled with its fidelity $\gamma_i\in [0,1]$.}
\label{fig:oneway}
\end{figure*}

SurfNet is based on this finding, and transfer the two parts, the Core and the Support, separately via two communication channels: the entanglement-based channel and the plain channel. In the following, we introduce the components within our quantum network, our proposed one-way communication procedure, and our new error correction decoder.

\subsection{Components}

In this section, we introduce the main components within SurfNet.

\textbf{Users } User nodes generate communication requests in the network, which are to transfer quantum messages from one user to another. 

\textbf{Optical Fibers } Optical fibers serve as edges in quantum networks that interconnect users and switches. Each optical fiber in SurfNet is equipped with two communication channels: an entanglement-based channel and a plain channel.

\textbf{Switches } Switch nodes in SurfNet serve as intermediate stations along each communication path. Each switch operates two routines, corresponding to the two communication channels. One routine continuously generates entangled qubit pairs, which are then shared with adjacent switches through the entanglement-based channels. The other routine re-encodes any incoming Support parts of surface codes from the plain channels into new photons, aiming to reduce the chance of random decoherence.

\textbf{Servers } Server nodes are essentially switch nodes with larger quantum memories. Similar to switches, servers operate the same two routines corresponding to the two communication channels. Additionally, servers are capable of performing error correction when a complete surface code arrives.

\subsection{One-way communication}

In SurfNet, all messages are encoded into surface codes before transmission. Transmitting each surface code involves the transmission of the Core part and the Support part. Notably, each part can be further subdivided into smaller parts to enable efficient parallel transmission. For ease of illustration, here we do not subdivide the Core part and Support part into smaller parts.

Transmission of the Core part is through the entanglement-based channel, which ensures low error rate via techniques of entanglement purification. Once the routing path is scheduled, all switches or servers along the path start to share its generated entangled pairs through the entanglement-based channels with its adjacent nodes along the path. Then, each intermediate node performs entanglement swapping to transform the chain of adjacent entanglements into a single entangled pair shared between the sender node and receiver node. 
This procedure of entanglement generation and swapping needs to be repeated for several times, then entanglement purification is deployed to combine these multiple low-quality entangled pairs into a single pair of high-quality entangled qubits. 
Finally, this entangled pair is consumed in quantum teleportation to transfer the Core part from the sender node to the receiver node.

Meanwhile, data qubits within the Support part are encoded into individual photons and physically transmitted through the plain channels. Since error correction requires the complete surface code to be present in the same server, the Support part(s) normally arrives earlier in the server and keeps being refreshed via error mitigation circuits until the Core part(s) arrives.

During transmission, some qubits may encounter erasure errors. When a qubit is erased, we substitute it with a maximally mixed state, initialized as $\ket{0}$ and randomly subjected to a Pauli gate chosen uniformly from $\{I, X, Y, Z\}$. Thus, on each transmitted surface code, the replaced qubits exhibit a much higher error rate.

Example one-way communication is illustrated in Fig.~\ref{fig:oneway}. Before the error correction in Server, the Core part is teleported to Server by consuming the ancillary entangled pair generated in Switch A and shared through the entanglement-based channel. Meanwhile, the Support part is physically transmitted as photons to Server through the plain channel, bypassing Switch B. After the error correction, both routes bypass Switch C, yet once again, the route for the Core part is through the entanglement-based channel, and the route for the Support part is through the plain channel. It is important to note that the two routes do not necessarily need to take the same path after error correction, and the routes in Fig.~\ref{fig:oneway} are just for illustrative purposes.

\subsection{Error correction decoder}\label{err_decoder}

Before we introduce our SurfNet Decoder, we firstly illustrate how mainstream decoders can be adapted to solve the error corrections in SurfNet. We begin by representing each surface code as a \textit{decoding graph} $G=\{V,E,W\}$, where vertices $V$ are measurement qubits, and edges $E$ are data qubits. Computing the weights $W$ needs some extra work. Notice in SurfNet, data qubits are transferred via different channels and different routes, so these data qubits are with different error rates and should be given different weights. 
The \textit{estimated fidelity} for each data qubit is calculated by multiplying the fidelity of all optical fibers it travels through, denoted as $\rho=\Pi_{i\in \text{path}} \gamma_i$, where $\gamma_i$ is the fidelity of the optical fiber. For a data qubit in the Core part, after each time of entanglement purification, its new $\rho = \frac{\rho_1\rho_2}{\rho_1\rho_2+(1-\rho_1)(1-\rho_2)}$~\cite{li2022fidelity}, for which $\rho_1$ and $\rho_2$ are the estimated fidelity of the consumed entangled pairs. For erasures, the estimated fidelity equals to $0.5$. The weight at each edge is then computed as $w=-\ln{(1-\rho)}$.

\begin{algorithm}[t]
\caption{Modified MWPM Decoder}\label{alg:blossom}
\textbf{Input:} set of Syndromes $\sigma$, estimated data qubit fidelity $\{\rho_i\}$
\\
\textbf{Output:} estimated error pattern $\Gamma$ 

\begin{algorithmic}[1]
\State Construct decoding graph $G=\{V, E, W\}$
\State Create path graph $G'=\{\sigma, E', W'\}$, initialize $G'$ as $\{\sigma, \{\}, \{\}\}$
\For{$u,v \in \sigma$ with $u\ne v$}
    \State Find $P(u,v)$, the shortest path from $u$ to $v$ in $G$ 
    \State Compute $w'(u,v)=\sum_{e \in P(u,v)} w_e$ 
    \State Add $(u,v)$ and $w'(u,v)$ to $E'$ and $W'$ of $G'$
\EndFor
\State Apply Blossom to $G'$ and get the minimum weight perfect matching $M'$
\State Initialize $\Gamma$ as empty
\For{$(u,v)$ in $M'$}
    \State add the edges of $P(u,v)$ to $\Gamma$
\EndFor
\end{algorithmic}
\end{algorithm}

\begin{theorem}
\label{the_mwpm}
Error corrections in SurfNet can be solved using a minimum weight perfect matching (MWPM) algorithm.
\end{theorem}

\begin{proof}
Following the above construction, let $G=\{V,E,W\}$ be an arbitrary decoding graph, and $\sigma\subseteq V$ be the set of syndromes. As mentioned previously, the error correction decoder aims to connect all syndromes into pairs, and each syndrome should be connected exactly once. Each connected path made by the decoder is associated with a \textit{likelihood}, which correlates with the fidelity of each data qubit along the connected path. And the performance of the decoder is evaluated by summing the likelihoods of all connected paths it made: the higher the likelihood is, the better the decoder performs. 
This decoding problem can be efficiently formulated as the minimum weight perfect matching problem. A \textit{perfect matching} $M$ of graph $G$ is defined as $M\subseteq E$ such that each vertex in $V$ is incident to one and only one edge in $M$. Notice in the decoding graph $G$, there exist vertices $V\setminus \sigma$ that are without syndromes and should not be connected. Thus, before we apply the minimum weight perfect matching algorithm, we need to construct a path graph $G'$. 
As in Algorithm~\ref{alg:blossom}, syndromes $\sigma$ constitute the vertices in $G'$, and they are interconnected via the shortest paths found in the original graph $G$. These paths are added as edges into $G'$. Then a minimum weight perfect matching algorithm (for example, as in Algorithm~\ref{alg:blossom}, the blossom algorithm~\cite{edmonds1973matching}) is applied to $G'$ to get the minimum weight perfect matching $M'$. The requirement of minimum weight ensures that such $M'$ represents the decoding result with a high likelihood. Notably, this result is not guaranteed to be with the maximal likelihood, since the minimum weight perfect matching formulation serves as an approximation to the actual decoding problem. The proof on the correctness of this approximation is accomplished as a joint result from~\cite{fowler2013minimum} and~\cite{korte2011combinatorial}, which handles the error correction for surface codes in two cases: surface codes without boundaries, and surface codes with boundaries.
\end{proof}

\begin{corollary}
\label{the_blo}
Error corrections in SurfNet can be solved via the sparse blossom algorithm in $O(n^2)$.
\end{corollary}

\begin{proof}
The proof follows as a direct result from Theorem~\ref{the_mwpm}, since the sparse blossom algorithm~\cite{higgott2023sparse} is a solver for the minimum weight perfect matching problem. Notably, different from other conventional minimum weight perfect matching decoders, the sparse blossom algorithm solves for the \textit{embedded matching} $M_E$, which is defined as $M_E\subseteq E$ such that each vertex in $\sigma$ is is incident to an odd number of edges in $M_E$, and each vertex in $V\setminus \sigma$ is incident to an even number of edges in $M_E$. The sparse blossom algorithm aims to find a minimum weight embedding matching of the decoding graph, and this problem proves in~\cite{higgott2023sparse} to be a reduction from the original minimum weight perfect matching problem. As shown in~\cite{higgott2023sparse}, the sparse blossom algorithm outputs results with an observed complexity of $O(n^{1.32})$. And updating the decoding graph with our calculated fidelity requires $O(n^2)$. Thus, the total time for error corrections in SurfNet using the sparse blossom algorithm is $O(n^2)$.
\end{proof}

So far, we have proven that conventional minimum weight perfect matching algorithms can be modified for solving the error corrections in SurfNet, achieving a $O(n^2)$ running time. In addition, as noted in~\cite{higgott2023sparse}, employing the sparse blossom algorithm in our Algorithm~\ref{alg:blossom} can be adapted for a parallel implementation, with a theoretical amortized complexity of linear time. However, employing the sparse blossom algorithm also suffers a worst time complexity of $O(nq^3+mq^2)$, in which $n$ is the number of vertices in the graph, $m$ is the number of edges, and $q$ is the number of syndromes.
\\

\begin{algorithm}[t]
\caption{SurfNet Decoder}\label{alg:ufdecoder}
\textbf{Input:} set of Syndromes $\sigma$, set of Erasures $\delta$, estimated data
\\ \hspace*{\algorithmicindent}\hspace*{\algorithmicindent} 
qubit fidelity $\{\rho_i\}$, decoder step size $r=2/3$
\\
\textbf{Output:} estimated error pattern $\Gamma$ 

\begin{algorithmic}[1]
\State Create list of odd clusters $\{C_1,C_2,\ldots\}$ that each contains a single syndrome
\While{exists odd cluster $C_i$}
    \For{all odd Clusters $\{C_i\}$}
        \State \textbf{Grow} $C_i$ by corresponding \textbf{speed}:
            \State \textbf{ } $-r/\ln{(1-0.5)}$ edge if on $\delta$
            \State \textbf{ } $-r/\ln{(1-\rho_i)}$ edge if else
        \State If $C_i$ meets another cluster, fuse together
        \State If $C_i$ becomes even, remove from odd list
    \EndFor
\EndWhile
\State Find a spanning tree for each cluster
\State Apply peeling decoder to the spanning forest to find $\Gamma$
\end{algorithmic}
\end{algorithm}

\begin{figure*}[t]
\centering
\fbox{\includegraphics[width= \linewidth]{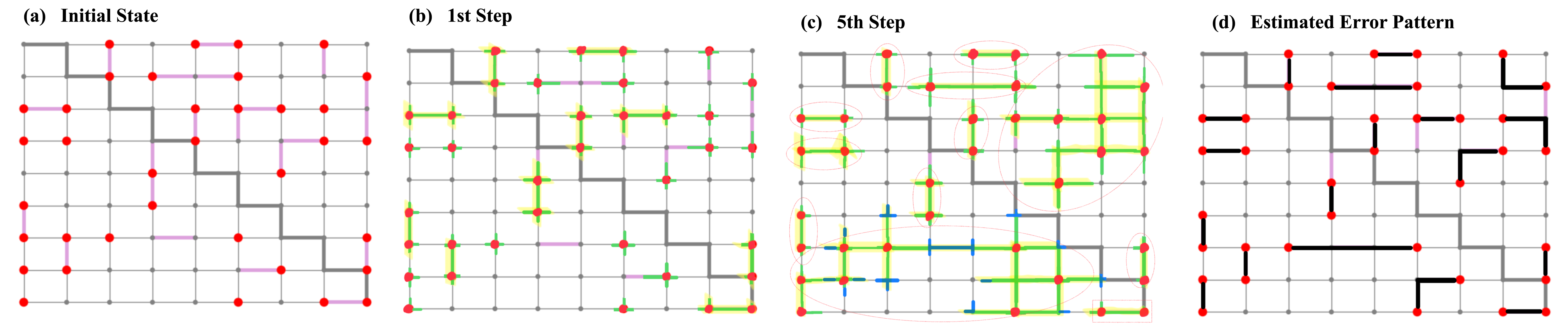}}
\caption{\text{Example decoding graph of the SurfNet Decoder}. Each vertex is a measurement qubit and each edge is a data qubit. For illustration, the growth speeds on Erasures, Core, Support are $\{1/2,1/8,1/4\}$. \textbf{(a)} \text{Initial State}. Red vertices denote Syndromes, Purple lines denote Erasures, and Solid Grey lines denote the Core part. \textbf{(b)} \text{Growth Step 1}. Green lines denote newly added edges for each cluster. Even clusters are marked in Yellow. \textbf{(c)} \text{Growth Step 5}. Blue lines denote newly added edges for each cluster, and Green lines denote existing edges. After this step, all clusters become even and finish growing. \textbf{(d)} \text{Estimated Error Pattern}, where Black lines denote data qubits with errors.}
\label{fig:ufdecoder}
\end{figure*}

Thus, it motivates us to also design a new decoder that works better in the worst case scenarios, which is our SurfNet Decoder. 
Overall, the SurfNet Decoder is based on the definition of ``growth speed" of clusters. Its routine is adapted from the Union-Find Decoder~\cite{delfosse2021almost} with the peeling decoder~\cite{delfosse2020linear}, and is tailored explicitly for our dual-channel network architecture. The algorithm is outlined in Algorithm~\ref{alg:ufdecoder}, with an illustrated procedure in Fig.~\ref{fig:ufdecoder}.
Among the inputs to the SurfNet Decoder, syndromes and erasures are derived from the measurement outcomes, and the estimated fidelity is computed as previously described. The decoder step size can be further adjusted to optimize between the decoding speed and accuracy, with the default $2/3$ generally achieving a good balance.
During the decoding, the initial step involves creating a singleton set for each syndrome vertex, denoted as a \textit{cluster}. Additionally, we define odd clusters to be clusters containing an odd number of syndromes, and even clusters to be clusters containing an even number of syndromes. In this step, each syndrome is an odd cluster and its cluster solely contains itself.
Next we grow the exteriors of each cluster at different speeds, calculated as $-r/\ln{(1-\rho_i)}$, where $r$ is the decoder step size and $(1-\rho_i)$ is the probability of errors occurring at the corresponding edge. The speed is maximized at erasures, as the replaced qubits with maximally mixed states exhibit the highest error rate.
When two expanded edges touch or one edge reaches the exterior of another cluster, the two clusters fuse together. If the combined cluster becomes an even cluster, we remove it from the list of odd clusters, as it no longer needs to grow.


\begin{theorem}
\label{the_surf}
Error corrections in SurfNet can be solved via the SurfNet Decoder with a worst case complexity of $O(n\alpha(n))$, where $\alpha$ is the inverse Ackermann function.
\end{theorem}

\begin{proof}
First we prove the near-maximal likelihood of the output $\Gamma$ from the SurfNet Decoder. Let $\Gamma'$ be arbitrary error pattern that is with a different likelihood from $\Gamma$. Denote $\{C_1,C_2,\ldots,C_k\}$ as the set of even clusters at the end of the SurfNet Decoder. As the optimality of the peeling decoder is proven in~\cite{delfosse2020linear}, then there exists at least one path in $\Gamma'$ that travels from one even cluster to another, which spans extra edges that were not grown by any clusters. Also by Algorithm~\ref{alg:ufdecoder}, each cluster suffices to be a shape with constant radius, with each point on its boundary reachable by the same likelihood, and any exterior edges can only be reached with a lower likelihood, which determines $\Gamma'$ to be with a lower likelihood than the algorithm output $\Gamma$. 
The worst time complexity of $O(n\alpha(n))$ is guaranteed by deploying an union-find implementation as in~\cite{delfosse2021almost}. In this implementation, each cluster is recorded as a tree data structure and its boundary vertices are recorded in a separate list. The boundary lists are updated at each round of growth, and cluster fusion is accomplished using the union-find algorithm. Several subroutines for tracing the states of each cluster are also employed, and the implementation details can be found in~\cite{delfosse2021almost}.
\end{proof}
    \textbf{}

\section{Network Routing}

Given the limited capacity and low entanglement generation rates at each switch and server, we aim to design a reliable and efficient network that achieves a high network throughput while maintaining a high successful rate for each communication. 
And we make the following assumptions:

\begin{itemize}
\item The network is bidirectional and connected, ensuring the existence of at least one path from one user to another.
\item The fidelity of each optical fiber can be measured and remains constant during routing.
\item Switch nodes, server nodes, and optical fibers have finite capacities that are not excessively large.
\end{itemize}
The routing procedure consists of two stages: scheduling and execution. 
To ensure simplicity and efficiency, we adopt offline scheduling to allocate available resources. And we employ online execution to promptly recover routing paths from errors or failures, minimizing the waste of utilized resources. 

\subsection{Offline scheduling}

First, we introduce the offline scheduling in SurfNet. Before each round of routing, the routing protocol collects requests from user nodes and operating information from switches, including their current capacity and entanglement status. For each request, the Support part consumes capacity in switches, and the Core part additionally consumes the prepared entangled pairs. The routing problem is then formulated as an optimization problem with capacity and entanglement resources acting as constraints. To provide a formal mathematical representation of our routing protocol, we define key terms as outlined in Table~\ref{table:notation_3}.
As listed in its last column, the variables in our formulation consist of four sets of integer variables: $Y_k$ to decide whether each request will be scheduled or partially scheduled, $a_e^k$ and $b_e^k$ to decide routes for each request, and $x_r^k$ to decide the number of scheduled error corrections for each request.

Computing fidelity for surface codes involves multiple exponentiations and multiplications, leading to a non-linear problem formulation. Thus, we translate the fidelity at each edge as its \textit{noise}. The noise at each edge is calculated as $\mu = \log(1/\gamma)$, where $\gamma$ is its fidelity measured in the range $[0,1]$. This allows noise to accumulate as summations instead of multiplications, and lower values are better. Correspondingly, we model the effect of each error correction as a decrease in the current noise to the surface code. And when computing the total noise to a surface code, we separately sum the noise for the Core and Support parts, to account for their different influences in logical error rates. 
For example, if we transmit a surface code of 25 data qubits, with 7 data qubits in the Core part, through the route in Fig.~\ref{fig:oneway}, then its two expected noises are: 
\begin{equation*}
\begin{aligned}
& \frac{7}{7}(\mu_1+\mu_2+ \mu_5+\mu_6)- \omega   \\
& \frac{7}{25}*\frac{1}{2}(\mu_1+\mu_2+ \mu_5+\mu_6)+\frac{18}{25}(\mu_3+\mu_4+ \mu_5+\mu_6) - \omega \\
\end{aligned}\end{equation*}
correspondingly for the Core part and the entire surface code. 
Notice in the second equation, when summing the noise for the entire surface code, the noise for the Core part is halved to account for the effect of entanglement purification through the entanglement-based channel. As later shown in Eq.~\ref{eq:c7}, we require both calculated noises to be below their corresponding thresholds to ensure the high fidelity of each communication.

Our integer programming for the routing problem states as follows. The objective function is to sum the number of concurrent communications, with the goal of maximizing the total throughput in our network:

\begin{equation}
\max \sum_{k\in \mathbb{K}} Y_k
\end{equation}
with the following constraints of Eq.~[\ref{eq:c1_3}]-[\ref{eq:c7}]:

\begin{equation}
\label{eq:c1_3}
\begin{aligned}
&\ Y_k\in [0,i_k] \textbf{ , } x^k_r\in [0,i_k] \textbf{ , } a^k_e\ge 0 \textbf{ , } b^k_e\ge 0 \\[1.5ex]
\end{aligned}
\end{equation}
which sets up the feasible regions of variables.

\begin{equation}
\begin{aligned}
&\ \sum_{e\in E} a^k_{(d,j)}+b^k_{(d,j)}+a^k_{(i,s)}+b^k_{(i,s)} = 0 & \forall_{k=(s,d)\textbf{ }\in \mathbb{K}} &\\ 
&\ \sum_{(i,d)\in E} a^k_{(i,d)}= \sum_{(s,j)\in E} a^k_{(s,j)} = n*Y_k & \forall_{k=(s,d)\textbf{ }\in \mathbb{K}} &\\
&\ \sum_{(i,d)\in E} b^k_{(i,d)}= \sum_{(s,j)\in E} b^k_{(s,j)} = m*Y_k & \forall_{k=(s,d)\textbf{ }\in \mathbb{K}} &\\
\end{aligned}
\end{equation}
which are the Initialization and Termination constraints that apply to the sender nodes and destination nodes. 

\begin{equation}
\begin{aligned}
&\ \frac{1}{n}\sum_{(i,r)\in E} a^k_{(i,r)} = \frac{1}{m}\sum_{(i,r)\in E} b^k_{(i,r)} = x^k_r & \forall_{ r\in\mathbb{R}\mathbb{R}, k\in \mathbb{K}} &\\
&\ \sum_{(i,r)\in E} a^k_{(i,r)} = \sum_{(r,j)\in E} a^k_{(r,j)} & \forall_{ r\in\mathbb{R}, k\in \mathbb{K}} &\\
&\ \sum_{(i,r)\in E} b^k_{(i,r)} = \sum_{(r,j)\in E} b^k_{(r,j)} & \forall_{ r\in\mathbb{R}, k\in \mathbb{K}} &\\
\end{aligned}
\end{equation}
which are the Conservation constraints that apply to the switch and server nodes.

\begin{equation}
\begin{aligned}
&\ \sum_{k\in\mathbb{K}} \sum_{(i,r)\in E} a^k_{(i,r)}+b^k_{(i,r)} \le \eta_r & \forall_{r\in\mathbb{R}} &\\
&\ \sum_{k\in\mathbb{K}} a^k_{(i,j)} + \sum_{k\in\mathbb{K}} a^k_{(j,i)} \le \eta_e & \forall_{e=(i,j)\in\mathbb{E}} &\\
\end{aligned}
\end{equation}
which are the Capacity constraints that ensure the computing resources within the network cannot be over-consumed. 
The first set of constraints sum the number of data qubits scheduled to be stored in each switch and server, and ensure they do not exceed the quantum memory capacity at each location. 
The second set of constraints sum the number of teleportation scheduled across each optical fiber, and ensure they do not exceed the number of prepared entangled pairs in the corresponding optical fiber.

\begin{table}[t]
\begin{center}
\caption{\label{table:notation_3}TABLE OF NOTATIONS USED IN ROUTING FORMULATION}
\begin{tabular}{|| l l ||} 
\hline\hline
 Notation & Definition \\ 
\hline\hline
 $\mathbb{U}$ & The set of all Users \\
 $\RR$ & The set of all Switches (including Servers) \\
 $\RR\RR$ & The set of all Servers \\
 $E$ & The set of all Edges \\
\hline
 $\eta_r$ & Storage capacity in each switch $r\in \RR$ \\
 $\eta_e$ & Number of prepared entanglements across each edge $e\in E$ \\
 $\mu_e $ & Noise at each edge $e$\\
\hline
 $\mathbb{K}$ & The set of all Communication Requests $k=[(s_k,d_k),i_k]$\\
 $i_k$ & Number of Surface Codes in Request $k$\\ 
 $n$ & Number of Core data qubits in each surface code\\
 $m$ & Number of Support data qubits in each surface code\\
\hline
 $\omega$ & Noise Reduction for performing error correction in a server \\
 $W_c$ & Noise Threshold for the Core part \\
 $W$ & Noise Threshold for the entire surface code\\
\hline
 $Y_k$ & Integer variable of determining Execution of request $k$\\
 $a^k_e$ & Integer variable of determining number of Core qubits for \\
 & request $k$ travelling through edge $e$ \\ 
 $b^k_e$ & Integer variable of determining number of Support qubits for \\
 & request $k$ travelling through edge $e$ \\ 
 $x^k_r$ & Integer variable of determining number of Error Corrections\\
 & scheduled for request $k$ in each server $r$\\
 \hline
\end{tabular}
\end{center}
\end{table}

\begin{equation}
\label{eq:c7}
\begin{aligned}
&\ 0\le \sum_{e\in E}(\mu_e*a^k_e) - \omega\sum_{r\in \mathbb{R}\mathbb{R}} x^k_r \le W_c*Y_k & \forall_{k\in \mathbb{K}} &\\
&\ \sum_{e\in E} [\mu_e*(\frac{1}{2} a^k_e+ b^k_e)] - \omega\sum_{r\in \mathbb{R}\mathbb{R}} x^k_r \le W*Y_k & \forall_{k\in \mathbb{K}} &\\
\end{aligned}
\end{equation}
which are the Noise constraints that ensure high-fidelity communication, computed as summing the noise along the routing path and subtracting the number of error corrections performed. 
The first set of constraints sum the noise only for the Core parts, to account for their critical influences in logical error rates of surface codes. The summation is required to be below the threshold $W_c$ to prevent excessive noise accumulation that can drift up the logical error rates. We also force this summation to be above 0 to prevent the inclusion of consecutive servers in a single route to conserve the computing resources.
The second set of constraints sum the noise for the entire surface code, and is required to be below the threshold $W$. 
As demonstrated later in Section~\ref{result}, fine-tuning the two thresholds $W_c$ and $W$ allows for balancing communication fidelity and network throughput.

Note that $\eta_r,\eta_e,\mu_e,i_k$ are collected from the network, and $n,m,\omega,W_c,W$ are pre-defined parameters, so the above constraints remain linear. 
However, as some of these variables are integers, the above integer programming formulation does not guarantee a polynomial-time solution. Feasible solutions can be adapted from existing algorithms for solving the integer minimum-cost Multi-Commodity Network Flow (MCNF) problems~\cite{salimifard2022multicommodity} using the Dantzig-Wolfe decomposition principle as in~\cite{hu1963multi, tomlin1966minimum}. For simplicity, in this chapter, we employ a relaxed Linear Programming version with rounding in the Evaluation section to simulate the performance of SurfNet.

\begin{figure*}[t]
\centering
\includegraphics[width= .9\linewidth]{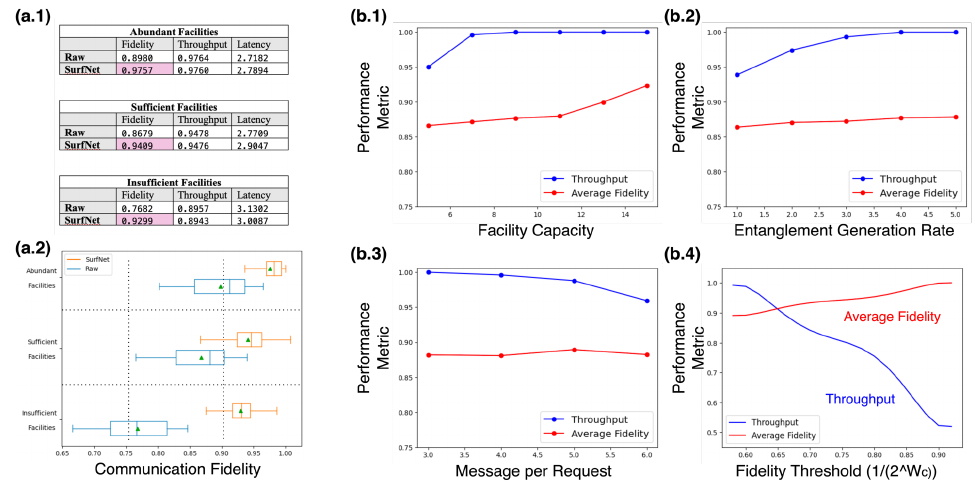}
\caption{\textbf{(a)} Comparison between Raw and SurfNet in different network scenarios. In each scenario, the comparison is over three evaluation metrics shown in table, and the comparison over communication fidelity is detailed in the plots. \textbf{(b)} Performance of SurfNet with respect to different \textbf{(b.1-3)} network parameters and \textbf{(b.4)} routing parameter. The performance is evaluated using two metrics: fidelity and throughput.}
\label{fig:eval_r_3}
\end{figure*}

\subsection{Online execution}

Next, we introduce the online execution in SurfNet. Upon receiving the routing schedules from the routing protocol, users directly transfer their Support parts of surface codes to the designated next node, while retaining the Core parts until reliable entanglements are established. Switches and servers operate two simultaneous routines: preparing entanglements for the Core parts and directly transferring the Support parts. In case an error correction is scheduled in a server, the server initiates a third routine: preparing error correction circuits and awaits the arrival of a complete surface code.

Common challenges arise from unpredictable situations like transmission delays or crashes in incoming/outgoing ports. And these issues are effectively addressed in our network. Each qubit has the option to remain in its current node and can be preserved in error mitigation circuits such as dynamic decoupling (DD) sequences. Additionally, if abundant resources are available in the local neighborhood, a node can locally replace a failed route with a recovery path leading to the next designated node. Thus, even without a centralized routing protocol, SurfNet can still operate in a hierarchical manner. In scenarios where only a few requests are generated in the network, centralized and hierarchical routing demonstrate similar performances.

Another main mechanism in our network is the parallelism between data transfer and error correction. Thanks to the intrinsic redundancy of surface codes, servers do not need to wait for all data qubits to be present before performing error correction. In SurfNet, error correction begins as soon as all Core parts and sufficient Support parts are collected, with all missing qubits marked as erasures. The tolerance for missing Support parts depends on their specific locations on the square lattice and is influenced by the overall fidelity of collected data qubits.
Generally, in SurfNet, the parallelism between data transfer and error correction presents another trade-off between efficiency and reliability: if the SurfNet Decoder always waits for the complete surface code to arrive, the reliability of error correction improves, but the message delivery efficiency is compromised.

In addition, we deploy opportunistic routing to accelerate the transmission through entanglement-based communication channels. Recall during the transmission of the Core parts, we need to wait for reliable entanglements to be established between the sender and receiver nodes. It requires multiple rounds of pairwise entanglements along the entire path, and its difficulty scales exponentially with the distance between the sender and receiver nodes. Thus, instead of waiting for the distant entanglement to be established, we allow the Core parts to move forward along the path as soon as a reliable entanglement is established across the next few nodes, even if the movement is only across a single optical fiber. This approach allows us to divide a long routing path into segments, significantly improving both communication efficiency and fidelity. Based on simulation experiments, we fix the minimum distance for the movement to be two consecutive optical fibers.
\\

\section{Performance Evaluation}

In this section, we conduct a series of simulations to evaluate the performance of the SurfNet network design and the SurfNet Decoder. We implement our proposed routing protocol and decoder algorithm for different network scenarios, and compare their performance against baseline and mainstream models.

\subsection{Evaluation methodology}

Considering the randomness in real-world network structures, we randomly generate networks without a predefined topology. Among these, three representative networks are chosen, each falling into a distinct scenario: a network with abundant facilities of switches and servers, one with sufficient facilities, and one with insufficient facilities. For each of the three networks, we randomly assign two fidelity to each optical fiber, representing scenarios of good or poor qualities of service. 
In addition, to evaluate the performance of the SurfNet Decoder, we implement it on surface codes with four different sizes: distance-9, distance-11, distance-13, and distance-15, for which distance refers to the minimum number of data qubits for a logical operator. Then the decoder is evaluated at different rates of Pauli errors, for the calculation of error threshold. 
We also include a fixed erasure error rate to mimic real scenarios within a network. Both erasure errors and Pauli errors randomly occur on each data qubit, without following a specific error pattern.

\subsection{Benchmarks}

Our proposed network design is denoted as \textbf{SurfNet}, and is compared to the following network designs:

\begin{itemize}
\item \textbf{Raw}: the baseline version of SurfNet, which does not divide surface codes into Core and Support parts, and all data qubits are transmitted through plain channels. Error corrections are available in each server, and all switches and servers have increased capacity as they no longer need to prepare entanglements.
\item \textbf{Purification N$=$ 1,2,9 }: the mainstream quantum networks that purely utilize quantum teleportation to transfer data qubits, and employ entanglement purification to improve communication fidelity. In a straightforward implementation, purification is carried out at each optical fiber in the network. Three scenarios are considered, denoted as \textbf{``N$=$ 1", ``N$=$ 2", ``N$=$ 9"}, where N indicates the number of additional entanglement pairs consumed for purification at each optical fiber.
\end{itemize}
The above network designs are tested in three network scenarios, which are generated using the Barabasi-Albert model with over 20 nodes, with the most connected nodes chosen to be the servers and switches. Fidelity at each optical fiber is randomly assigned within the range of $[0.75, 1]$ representing good-connection scenarios and within $[0.5, 1]$ representing poor-connection scenarios. Each network design undergoes 1080 trials with different sets of network parameters, including the capacity for switches and servers, entanglement generation rate, number of requests, and number of messages in each request.

Our proposed error correction decoder is denoted as the \textit{SurfNet Decoder}, and is compared against the baseline \textit{Union-Find Decoder} implemented as in~\cite{delfosse2021almost}. Each input surface code is subjected to both erasure errors and Pauli errors, with the erasure rate fixed at 15\%, and the Pauli error rate ranging between 5.0-8.5\%. These error rates are halved at the Core part of each surface code.

\subsection{Results}\label{result}

In this section, evaluation is conducted over three evaluation metrics: \textit{fidelity}, \textit{latency}, and \textit{throughput}. Fidelity and latency are defined as the success rate and waiting time for each communication, and are computed as the average of all communications executed in the network. Throughput is calculated as the number of executed communications divided by the total number of requested communications. 

From the tables (a.1) in Fig.~\ref{fig:eval_r_3}, throughput and latency in both networks are similar, while SurfNet exhibits much higher average communication fidelity in all three network scenarios as in (a.2). This improvement is mainly due to the dual-channel characteristic of SurfNet, and the partitioning of each surface code into the Core and the Support parts proves to be effective in maintaining communication fidelity.

Fig.~\ref{fig:eval_r_3} (b.1-3) provide a detailed analysis of SurfNet with respect to different network parameters. Facility capacity in (b1) and entanglement generation rate in (b2) are the two main parameters controlling the total computing resources within the network. Their direct impacts on both throughput and fidelity are evident since, naturally, more resources can support more communications. Message per request in (b3) limits the maximum number of surface codes allowed in each request, thus controlling the consumption rate of the computing resources. From the graph, it is observed that it has little impact on fidelity but seemingly affects throughput, as the supply of computing resources can no longer meet the growing demand.
And in Fig.~\ref{fig:eval_r_3} (b.4), we analyze the effect of fidelity threshold in the routing protocol of SurfNet, computed as $1/(2^{W_C})$ where $W_C$ is the noise threshold. As mentioned earlier, the noise thresholds in SurfNet can be used to strike for a balance between communication fidelity and network throughput. As in the graph, a low fidelity threshold allows for low-fidelity communications, resulting in high network throughput but low average fidelity for each communication. Conversely, a higher fidelity threshold is more selective for communication quality, leading to lower network throughput but higher average fidelity.

\begin{figure}[t]
\centerline{\includegraphics[width= \linewidth]{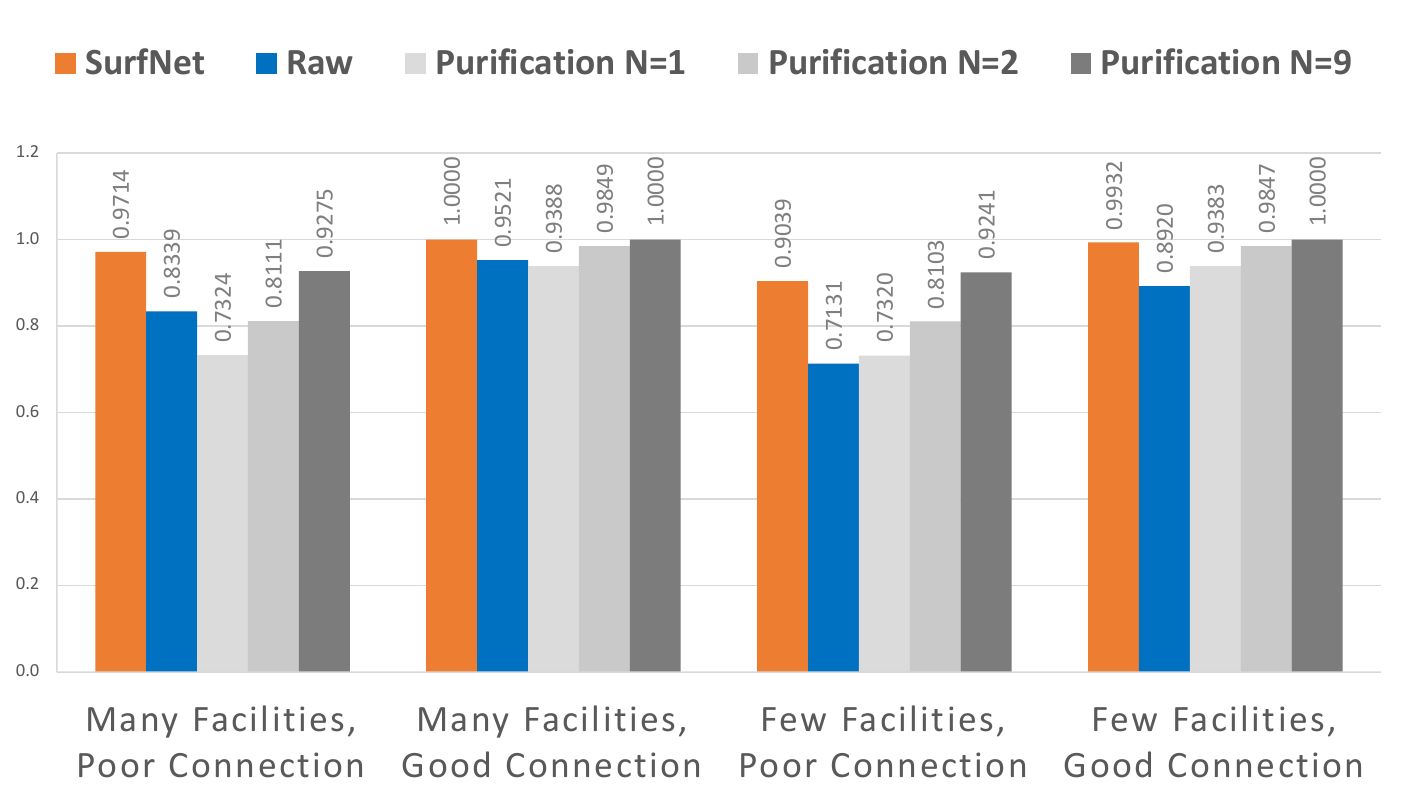}}
\caption{\text{Comparison on averaged communication fidelity} between the five routing designs in four network scenarios. In each scenario, the result is computed as the average from all 1080 trials.}
\label{fig:eval_r_3r}
\end{figure}

An overall comparison of all five network designs is shown in Fig.~\ref{fig:eval_r_3r}. We configure the routing protocols in all networks to yield similar throughputs, and focus only on the average fidelity in each network design. SurfNet achieves high fidelity across all four scenarios. It demonstrates significant advantage in networks with abundant facilities, while its drawbacks become more apparent in scenarios with limited facilities and poor connections. Improving the online execution stage of SurfNet, such as finding better recovery paths or incorporating adaptive code sizes based on quality of service, are potential directions for enhancing SurfNet in such scenarios.

Fig.~\ref{fig:eval_de} compares error correction performance between the Union-Find decoder and the SurfNet Decoder. In general, below the pauli error threshold, a surface code with a larger size tends to have a lower logical error rate. Conversely, above the pauli error threshold, a surface code with a smaller size tends to have a lower logical error rate. The SurfNet Decoder exhibits a higher threshold of 7.25\% compared with the 7.1\% of Union-Find decoder. This advantage can be further enhanced if the Core part of the surface code is configured to be larger or has a more optimized geometry, which presents potential future directions for improving the SurfNet Decoder.
\\

\begin{figure}[t]
\centerline{\includegraphics[width= \linewidth]{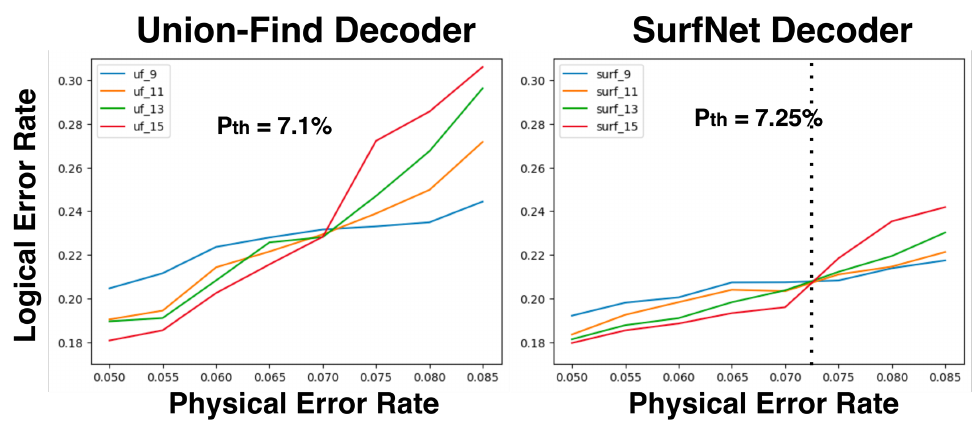}}
\caption{\text{Pauli error threshold} of surface codes using Union-Find decoder (left) and SurfNet Decoder (right).}
\label{fig:eval_de}
\end{figure}

\section{Conclusion}

In this paper, we introduced our SurfNet, an innovative quantum network utilizing surface codes as logical qubits for quantum message transmission. We proposed a novel end-to-end communication procedure employing two parallel communication channels within SurfNet: the entanglement-based channel and the plain channel. This dual-channel approach effectively addresses the limitation posed by the low entanglement generation rate at each node, and offers a way to balance the reliability and efficiency of quantum networks. To enhance communication fidelity, error corrections can be performed at servers along the routing path. The routing protocol in our network is deigned to coordinate the balance between the network throughput and average communication fidelity. Additionally, we introduced the SurfNet Decoder, which can fully leverage the modular characteristic of surface codes within our network. For evaluation, we conducted simulations for various network scenarios, comparing SurfNet against its baseline and other mainstream quantum networks. We also separately compared our SurfNet Decoder against its baseline Union-Find decoder. The results demonstrated that both SurfNet and its decoder can significantly improve the overall fidelity of quantum networks.

\section*{Acknowledgments}
The authors would like to thank all the reviewers for their
helpful comments. This work was  supported in part by the Commonwealth Cyber Initiative (CCI, cyberinitiative.org).

\bibliographystyle{ieeetr}
\bibliography{sfnet}

\end{document}